\documentclass[a4paper,11pt,leqno]{article}
\usepackage{amsmath}
\usepackage{amsfonts}
\usepackage{eucal}
\usepackage{amsthm}
\usepackage[dvips,colorlinks=true,linkcolor=blue]{hyperref}
\usepackage{graphicx,subfigure}
\usepackage[rflt]{floatflt}
\numberwithin{equation}{section}

\newcommand{\R}{{\mathbb R}}

\newcommand{\Z}{{\mathbb Z}}

\newcommand{\re}{{\rm Re}\,}

\newcommand{\bigpare}[1]{\bigl(#1\bigr)}
\newcommand{\biggpare}[1]{\biggl(#1\biggr)}

\newcommand{\biggbra}[1]{\biggl\{#1\biggr\}}

\newcommand{\bigset}[2]{\bigl\{#1\bigm|#2\bigr\}}

\newcommand{\norm}[1]{\| #1 \|}

\newcommand{\jap}[1]{\langle #1 \rangle}

\def\a{\alpha}
\def\b{\beta}
\def\c{\gamma}
\def\d{\delta}
\def\e{\varepsilon}
\def\f{\varphi}
\def\g{\psi}

\def\k{\kappa}
\def\l{\lambda}
\def\m{\mu}
\def\n{\nu}
\def\s{\sigma}

\newcommand{\C}{\Gamma}
\newcommand{\F}{\Phi}
\newcommand{\G}{\Psi}

\renewcommand{\O}{\Omega}

\def\re{\mathbb{R}}

\def\ze{\mathbb{Z}}

\def\pa{\partial}

\newtheorem{thm}{Theorem}[section]
\newtheorem{lem}[thm]{Lemma}

\newtheorem{Le}{Lemma}[section]
\newtheorem{Pro}{Proposition}[section]

\theoremstyle{definition}

\newtheorem{ass}{Assumption}

\theoremstyle{remark}
\newtheorem{rem}[thm]{Remark}

\title{Lifshitz Tails for Generalized Alloy Type \\
  Random Schr{\"o}dinger Operators}

\author{Fr{\'e}d{\'e}ric KLOPP\footnote{%
    LAGA, U.M.R. 7539 C.N.R.S, Institut Galil{\'e}e, Universit{\'e} de
    Paris-Nord, 99 Avenue J.-B.  Cl{\'e}ment, F-93430 Villetaneuse,
    France\ et \ Institut Universitaire de France.  Email:
    \href{mailto:klopp@math.univ-paris13.fr}{\tt
      klopp@math.univ-paris13.fr}} \text{ } and Shu
  NAKAMURA\footnote{%
    Graduate School of Mathematical Sciences, University of Tokyo,
    3-8-1 Komaba, Meguro-ku, Tokyo, Japan 153-8914.  Email:
    \href{mailto:shu@ms.u-tokyo.ac.jp}{\tt shu@ms.u-tokyo.ac.jp} }}
\date{}
\begin{document}
\maketitle
\begin{abstract}
  We study Lifshitz tails for random Schr{\"o}dinger operators where the
  random potential is alloy type in the sense that the single site
  potentials are independent, identically distributed, but they may
  have various function forms.  We suppose the single site potentials
  are distributed in a finite set of functions, and we show that under
  suitable symmetry conditions, they have Lifshitz tail at the bottom
  of the spectrum except for special cases. When the single site
  potential is symmetric with respect to all the axes, we give a
  necessary and sufficient condition for the existence of Lifshitz
  tails. As an application, we show that certain random displacement
  models have Lifshitz singularity at the bottom of the spectrum, and
  also complete the study of continuous Anderson type models
  undertaken in~\cite{KN1}.
\end{abstract}


\section{Introduction}
\label{sec:introduction}
Consider the continuous alloy type (or Anderson) random Schr{\"o}dinger
operator:
\begin{equation}
  \label{eq:5}
  H_\omega=-\Delta+V_0+V_\omega\text{ where
  }V_\omega(x)=\sum_{\gamma\in\Z^d}\omega_\gamma V(x-\gamma)
\end{equation}
on $\R^d$, $d\geq 1$, where
\begin{itemize}
\item $V_0$ is a periodic potential;
\item $V$ is a compactly supported single site potential;
\item $(\omega_\gamma)_{\gamma\in\Z^d}$ are independent identically
  distributed random coupling constants.
\end{itemize}
Let $\Sigma$ be the almost sure spectrum of $H_\omega$ and
$E_-=\inf\Sigma$. When $V$ has a fixed sign, it is well known that the
$E_-=\inf(\sigma(-\Delta+V_{\overline{b}}))$ if $V\leq0$ and
$E_-=\inf(\sigma(-\Delta+V_{\overline{a}}))$ if $V\geq0$. Here,
$\overline{x}$ is the constant vector
$\overline{x}=(x)_{\gamma\in\Z^d}$.\\
Moreover, for $E$ a real energy, one defines the {\it integrated
  density of states} by
\begin{equation}
  \label{eq:6}
  N(E)=\lim_{L\to+\infty}\frac{\#\{\text{eigenvalues of
    }H_{\omega,L}^N\ \leq E\}}{L^d}
\end{equation}
where
\begin{equation}
  \label{eq:9}
  H_{\omega,L}^N =-\triangle +V_0+V_\omega \qquad \text{on } L^2(C_L(0))
\end{equation}
with Neumann boundary conditions, 
where $C_L(0)$ is defined by \eqref{eq:12}. 
It is well-known that $N(E)$ exists
and is non-random, i.e., $N(E)$ is independent of $\omega$, almost
surely; it has been the object of a lot of studies.\\
In particular, it is well known that the integrated density of states
of the Hamiltonian admits a Lifshitz tail near $E_-$, i.e.,
\begin{equation*}
  \lim_{E\to E_-^+}\frac{\log|\log N(E)|}{\log(E-E_-)}<0.
\end{equation*}
Actually, the limit can often be computed and in many cases is equal
to $-d/2$; we refer to~\cite{CL, K1, K2, FP,MR1935594,V1, V2} for
extensive reviews and more precise statements.
\vskip.2cm In the present paper, we mainly consider a generalized
Bernoulli alloy type model that we define below: we allow the single
site potential to have various function forms (with a discrete
distribution). We give a necessary and sufficient condition to have
Lifshitz tail under a symmetry assumption on the single site
potentials.  The results we obtain are then applied to the random
displacement models studied recently by Baker, Loss and Stolz
(\cite{BLS1, BLS2}), and also to complete the study of the occurrence
of Lifshitz tails for alloy type models initiated in~\cite{KN1}.
\subsection{The model}
\label{sec:model}
Let us now describe our model. We let $d\geq 1$ and we
study operators on $\mathcal{H}=L^2(\re^d)$.  We denote
\begin{equation}
  \label{eq:12}
  C_\ell(x)=\bigset{y\in\re^d}{0\leq y_j-x_j\leq \ell, j=1,\dots,d}
\end{equation}
be the cube with the size $\ell>0$ and $x$ as a corner. Let $V_0\in
C^0(\re^d)$ be a background potential, which is periodic with respect
to $\ze^d$.

Let $v_k\in C_c^0(C_1(0))$, $k=1,\dots, M$, be single site potentials
where $M\in\mathbb{N}$. We consider the random Schr{\"o}dinger operator:
\begin{equation*}
H_\omega = -\triangle +V_0 +V_w \qquad \text{on } \mathcal{H}=L^2(\re^d),
\end{equation*}
where
\begin{equation*}
V_\omega(x) = \sum_{\c\in\ze^d} v_{\omega(\c)}(x-\c)
\end{equation*}
is the random potential and $\bigset{\omega(\c)}{\c\in\ze^d}$ are
independent identically distributed (i.i.d.) random variables with
values in $\{1,\dots, M\}$.

To fix ideas, let us assume
\begin{equation}
  \label{eq:1}
  \inf \s(H_\omega) =0, \quad \text{a.s. }\omega  
\end{equation}
which can always be achieved by shifting $V_0$ by a constant.

We denote
\begin{equation*}
H_k^N = -\triangle + V_0 + v_k \quad \text{on } L^2(C_1(0))
\end{equation*}
with Neumann boundary conditions on the boundary $\pa C_1(0)$.

Define
\begin{ass}
  (1) $V_0$ is symmetric about the plane $\bigset{x}{x_d=1/2}$.\\
  (2) There exists $m\in\{1,\dots,M\}$ such that
  \begin{equation*}
  \inf \s(H_k^N)=0 \qquad\text{for }k=1,\dots, m,
  \end{equation*}
  and
  \begin{equation*}
  \inf \s(H_k^N)>0 \qquad \text{for } k>m.
  \end{equation*}
  (3) Moreover, for $k=1,\dots, m$, $v_k(x)$ is symmetric about
  $\{x_d=1/2\}$.
\end{ass}
\begin{rem}
Note that in this assumption, we only require symmetry with
respect to a single coordinate hyperplane that we chose to be the
$d$-th one.

If one assumes that $V_0$ and the $(v_k)_{1\leq k\leq M}$ are
reflection symmetric with respect to all the coordinate planes (see
e.g.~\cite{BLS1,BLS2,KN1}), the standard characterization of the
almost sure spectrum (see e.g.~\cite{FP,K1}) and lower bounding
$H_\omega$ by the direct sum of its Neumann restrictions to the cubes
$(C_1(\gamma))_{\gamma\in\ze^d}$ show that, as a consequence
of~\eqref{eq:1}, one obtains
\begin{itemize}
\item for all $m\in\{1,\dots,M\}$, $\inf \s(H_k^N)\geq0$;
\item there exists $m\in\{1,\dots,M\}$ such that $\inf \s(H_m^N)=0$.
\end{itemize}
\end{rem}
\subsection{The results}
\label{sec:results}
We study the Lifshitz singularity for the integrated
density of states (IDS) at the zero energy. Recall that the IDS is
defined by~\eqref{eq:6}\\
We first consider a relatively easy case:
\begin{thm}
  \label{th:1}
  Suppose Assumption~A with $m<M$. Then
  \begin{equation}\label{LT}
    \limsup_{E\to+0} \frac{\log|\log N(E)|}{\log E} \leq -\frac{1}{2}.
  \end{equation}
\end{thm}

We expect \eqref{LT} holds with $-{d}/{2}$ in the right hand side,
which is known to be optimal (see e.g Theorem 0.2 and Section 2.2
in~\cite{KN1}).

If $m=M$, then we need further classification of the potential
functions.  We denote the standard basis of $\re^d$ by
\begin{equation*}
\mathbf{e}_j = (\d_{ji})_{i=1}^d \in\re^d, \quad j=1,\dots, d,
\end{equation*}
and we define an operator $H_{k\ell(j)}^N$ on $L^2(U_j)$ as 
\begin{equation}
  \label{eq:2}
  U_j =C_1(0)\cup C_1(\mathbf{e}_j), \quad j=1,\dots, d.  
\end{equation}
We set
\begin{equation}
  \label{eq:3}
  H_{k\ell(j)}^N =\begin{cases} -\triangle +V_0(x) +v_k(x) \quad
    &\text{on }C_1(0) \\
    -\triangle +V_0(x)+v_\ell(x-\mathbf{e}_j) \quad &\text{on }
    C_1(\mathbf{e}_j)
  \end{cases}
\end{equation}
with Neumann boundary conditions on $\pa U_j$, where
$k,\ell\in\{1,\dots,m\}$ and $j\in\{1,\dots,d\}$. We define
\begin{equation}
  \label{eq:14}
  v_j\underset{j}{\sim} v_\ell \quad \overset{def}{\Longleftrightarrow}
  \quad \inf\s(H^N_{k\ell(j)})=0.
\end{equation}
Namely, $v_k\underset{j}{\sim} v_\ell$ implies the coupling of two
local Hamiltonians $H_k^N$ and $H_\ell^N$ does not increase the ground
state energy. We note that $v_k\underset{j}{\not\sim} v_\ell$
generically for $k\neq \ell$.
\begin{thm}
  \label{th:2}
  Suppose Assumption~A with $m=M$. Suppose moreover that
  $v_k\underset{d}{\not\sim} v_\ell$ for some $k\neq \ell$.  Then
  \eqref{LT} holds, i.e., $H_\omega$ has Lifshitz singularities at the
  zero energy.
\end{thm}

In order to obtain a more precise result on the existence and the
absence of Lifshitz singularities, we make a stronger symmetry
assumption on the potentials.
\begin{ass}
  In addition to satisfying Assumption~A, $V_0$ and $v_k$ are
  symmetric about $\bigset{x}{x_j=1/2}$ for all $j=1,\dots, d$, and
  $k=1,\dots, m=M$.
\end{ass}
\begin{thm}
  \label{th:3}
  Suppose Assumption~B. Then
  \begin{enumerate}
    \renewcommand{\theenumi}{\roman{enumi}}
    \renewcommand{\labelenumi}{\textbf{{\rm(\theenumi)} }}
  \item If $v_k\underset{j}{\not\sim} v_\ell$ for some $j$ and $k\neq
    \ell$, then \eqref{LT} holds.
  \item If $v_k\underset{j}{\sim} v_\ell$ for all $j$ and $k,\ell$,
    then the van Hove property holds, namely, there exists $C>0$ such that
    \begin{equation}
      \label{vH}
      \frac1C E^{d/2}\leq N(E)\leq C E^{d/2}.
    \end{equation}
  \end{enumerate}
\end{thm}
\noindent In~\eqref{vH}, the asymptotic is new only for $E$ small; for
$E$ large, it is a consequence of Weyl's law. The example in
Section~3 of \cite{KN1} is a special case of (ii) of
Theorem~\ref{th:3}.

In a previous paper \cite{KN1}, we used the concavity of the ground
state energy with respect to the random parameters, and also used an
operator theoretical trick to reduce the problem to monotonous
perturbation case. These methods are not available under the
assumptions of the present paper. Instead, we employ a quadratic
inequality similar to the Poincar{\'e} inequality, and take advantage of
the positivity of certain Dirichlet-to-Neumann operators to obtain a
lower bound of the ground state energy for Schr{\"o}dinger operators on a
strip. This estimate is quasi one dimensional, and this is why we
obtain Lifshitz tail estimate with the exponent corresponding to one
dimensional case. We do believe that this method can be refined to
obtain the optimal exponent, though we have not been successful so
far.

This paper is organized as follows. We discuss the eigenvalue estimate
on a strip in Section~\ref{sec:lower-bounds-ground} and prove our main
theorems in Section~\ref{sec:proof-main-theorems}. We discuss an
application to random displacement models in
Section~\ref{sec:appl-rand-displ}, and an application to the model
studied in~\cite{KN1} in Section~\ref{sec:alloy-type-model}.

Throughout this paper, we use the following notations:
$\mathbb{P}(\cdot)$ denotes the probability measure for the random
potential, and $\mathbb{E}(\cdot)$ denotes the expectation;
$\mathcal{D}(A)$ denotes the definition domain of an operator $A$;
$\jap{\cdot,\cdot}$ denotes the inner product of $L^2$-spaces; $\pa
\O$ denotes the boundary of a domain $\O$; and $\#\Lambda$ denotes the
cardinality of a set $\Lambda$.
\vskip.1cm{\bf Acknowledgment:} It is a pleasure to thank Michael Loss and
Gunter Stolz for valuable discussions. We also thank the organizers of
the workshop ``Disordered Systems: Random Schr{\"o}dinger Operators and
Random Matrices'' at the MF Oberwolfach where part of this work was
done.

\section{Lower bounds on the ground state energy}
\label{sec:lower-bounds-ground}
Throughout this section, we suppose $v_1,\dots v_m$ satisfy
Assumption~A.  Let $a>0$,
\begin{equation*}
\O_0 =[0,1]^{d-1}\times [-a,0]\ \subset \re^d,
\end{equation*}
and let $W_0\in C^0(\O_0)$ be a real-valued function on $\O_0$.  We
set
\begin{equation*}
P_0^N =-\triangle +W_0 \quad\text{on }L^2(\O_0)
\end{equation*}
with Neumann boundary conditions. Let $L\in\mathbb{N}$,
\begin{equation*}
\O_1=[0,1]^{d-1}\times [0,L]
\end{equation*}
and let $W_1\in C^0(\O_1)$ such that
\begin{equation*}
W_1= V_0+ v_{k(\ell)}(x-\ell \mathbf{e}_d) \quad \text{if } x\in
C_1(\ell \mathbf{e}_d), \ \ell=0,\dots, L-1,
\end{equation*}
where $\{k(\ell)\}_{\ell=0}^{L-1}$ is a sequence with values in
$\{1,\dots, m\}$.  We then set
\begin{equation*}
\O=\O_0\cup\O_1, \quad W(x)=\begin{cases} W_0(x) \quad &\text{if } x\in \O_0 \\
  W_1(x) \quad &\text{if }x\in \O_1 \end{cases}
\end{equation*}
and set
\begin{equation*}
P^N =-\triangle +W \quad \text{on } L^2(\O)
\end{equation*}
with Neumann boundary conditions. Then, the main result of this
section is as follows.
\begin{thm}
  \label{th:4}
  Suppose $\inf \s(P^N_0)>0$, and suppose
  $v_{k(\ell)}\underset{d}{\sim} v_{k(\ell')}$ for $\ell,\ell'\in
  \{0,\dots,L-1\}$.  Then, there exists $C>0$ such that $C$ is
  independent of $L$ and of the sequence $\{k(\ell)\}$, and such that
  \begin{equation*}
  \inf \s(P^N) \geq \frac{1}{CL^2}.
  \end{equation*}
\end{thm}

In the following, we suppose $v_k\underset{d}{\sim} v_\ell$ for all
$k,\ell$ for simplicity (and without loss of generality).  We prove
Theorem~\ref{th:4} by a series of lemmas. First, we show a variant of
the classical Poincar{\'e} inequality.  Let $\C$ be the trace operator
from $H^1(\O_1)$ to $L^2(S)$ with $S=[0,1]^{d-1}\times\{0\}$, i.e.,
\begin{equation*}
\C\f(x') =\f(x',0) \quad \text{for }x'\in [0,1]^{d-1}, \ \f\in
C^0(\O_1),
\end{equation*}
and $\C$ extends to a bounded operator from $H^1(\O_1)$ to $L^2(S)$.
\begin{lem}
  Let $\f\in H^1(\O_1)$. Then
  \begin{equation*}
  \frac{2}{L} \norm{\C \f}^2_{L^2(S)} +\norm{\nabla\f}_{L^2(\O_1)}^2
  \geq \frac{1}{L^2} \norm{\f}_{L^2(\O_1)}^2.
  \end{equation*}
\end{lem}
\begin{proof}
  It suffices to show the estimate for $\f\in C^1(\O_1)$. Since
  \begin{equation*}
  \f(x',t)=\f(x',0)+\int_0^t \pa_{x_d}\f(x',s)ds, \quad x'\in
  [0,1]^{d-1}, t\in [0,L],
  \end{equation*}
  we have
  \begin{align*}
    |\f(x',t)| &\leq |\f(x',0)|+\int_0^t |\pa_{x_d}\f(x',s)|ds \\
    &\leq |\f(x',0)| +\sqrt{t} \biggpare{\int_0^L
      |\nabla\f(x',s)|^2ds}^{1/2}
  \end{align*}
  by the Cauchy-Schwarz inequality. This implies
  \begin{align*}
    \norm{\f}_{L^2(\O_1)}^2 &\leq \int \int_0^L \biggbra{|\f(x',0)| \
      +\sqrt{t} \biggpare{\int_0^L |\nabla\f(x',s)|^2ds}^{1/2}}^2 dt dx' \\
    &\leq 2 \int\int_0^L |\f(x',0)|^2 dsdx'
    +2 \int_0^L tdt \times \norm{\nabla\f}^2_{L^2(\O_1)}\\
    &= 2L\norm{\C\f}_{L^2(S)}^2 +L^2 \norm{\nabla\f}_{L^2(\O_1)}^2
  \end{align*}
  and the claim follows.
\end{proof}

For $k\in\{1,\dots, M\}$, we set
\begin{equation*}
q_k(\f,\g) =\int_{C_1(0)} \bigpare{\nabla\f\cdot\nabla\overline{\g} +
  v_k \f\overline{\g}} dx, \quad \f,\g\in H^1(C_1(0)),
\end{equation*}
which is the quadratic form corresponding to $H_k^N$.  Let $\G_k$ be
the positive ground state for $H_k^N$, which is unique up to a
constant. Since $\inf \s(H_k^N)=0$, we expect $\f/\G_k$ is close to a
constant if $q_k(\f,\f)$ is close to 0, and this observation is
justified by the following lemma.
\begin{lem}
  There exists $c_1>0$ such that
  \begin{equation*}
  \norm{\nabla(\f/\G_k)}_{L^2(C_1(0))}^2 \leq c_1 q_k(\f,\f), \quad
  \f\in H^1(C_1(0)), k=1,\dots, m.
  \end{equation*}
\end{lem}
\begin{proof}
  This is a consequence of the so-called {\em ground state transform}.
  It suffices to show the inequality when $\f\in C^1(C_1(0))$. We set
  $f=\f/\G_k$.  Then we have
  \begin{align*}
    q_k(\f,\f) &= \jap{\nabla(f\G_k),\nabla(f\G_k)} + \jap{v_k f\G_k,f\G_k} \\
    &= \norm{\G_k(\nabla f)}^2 + \jap{\G_k\nabla f, f \nabla \G_k}
    +\jap{f\nabla \G_k, \G_k \nabla f} \\
    &\hspace{3cm}  +\jap{f\nabla\G_k,f \nabla \G_k} + \jap{v_k f\G_k,f\G_k} \\
    &= \norm{\G_k(\nabla f)}^2 +\jap{\nabla(|f|^2\G_k),\nabla \G_k}
    +\jap{v_k |f|^2\G_k,\G_k} \\
    &= \norm{\G_k(\nabla f)}^2 + q_k(|f|^2\G_k,\G_k).
  \end{align*}
  Since $q_k(|f|^2\G_k,\G_k) =
  \jap{(H_k^N)^{1/2}|f|^2\G_k,(H_k^N)^{1/2}\G_k}=0$, we have
  \begin{equation*}
  q_k(\f,\f) = \norm{\G_k(\nabla f)}^2 \geq (\inf |\G_k|)^2
  \norm{\nabla f}^2,
  \end{equation*}
  and we may choose $c_1= (\min_k\inf|\G_k|)^{-2}$.
\end{proof}
\begin{lem}
  Suppose $v_k\underset{d}{\sim} v_\ell$. Then, there exists
  $\m_1,\m_2>0$ such that
  \begin{equation*}
  \m_1 \G_k(x',0)=\m_2 \G_\ell(x',0), \qquad \text{for }x'\in
  [0,1]^{d-1}.
  \end{equation*}
\end{lem}
\begin{proof}
  Consider $H_{k\ell(d)}^N$ in $U_d$ (see~(\ref{eq:2})
  and~(\ref{eq:3}) in Section~\ref{sec:introduction}), and let $\F\in
  L^2(U_d)$ be the positive ground state of $H_{k\ell(j)}^N$. We set
  \begin{equation*}
  \f_1 =\F\lceil_{C_1(0)}, \quad \f_2(\cdot) =
  \F(\cdot+\mathbf{e}_d)\lceil_{C_1(0)}.
  \end{equation*}
  Then $\f_1,\f_2$ are positive and $q_k(\f_1,\f_1)
  =q_\ell(\f_2,\f_2)=0$.  By the variational principle and the
  uniqueness of the ground states, we learn
  \begin{equation*}
  \f_1=\m_1\G_k, \quad \f_2=\m_2\G_\ell
  \end{equation*}
  with some $\m_1,\m_2>0$. By Assumption~A, $\G_k$ and $\G_\ell$ are
  symmetric about $\{x_d=1/2\}$, and hence
  \begin{equation*}
  \m_1\G_k(x',0) =\m_1\G_k(x',1)=\f_1(x',1)=\f_2(x',0)
  =\m_2\G_\ell(x',0)
  \end{equation*}
  for $x'\in [0,1]^{d-1}$, where we have used the continuity of $\F$
  on $\{x_d=1\}$.
\end{proof}
Now, let $\O_1$ and $W_1$ be as in the beginning of
Section~\ref{sec:lower-bounds-ground}, and define
\begin{equation*}
P_1^N=-\triangle +W_1 \quad \text{on } L^2(\O_1)
\end{equation*}
with Neumann boundary conditions. We set
\begin{equation*}
Q_1(\f,\g) =\int_{\O_1} \bigpare{\nabla\f\cdot\nabla\overline{\g}
  +W_1\f\overline\g} dx = \jap{(P_1^N)^{1/2}\f,(P_1^N)^{1/2}\g}
\end{equation*}
for $\f,\g\in H^1(\O_1)=\mathcal{D}((P_1^N)^{1/2})$. Then, we have
\begin{lem}
  There exists $c_2>0$ such that $c_2$ is independent of $L$ and of
  the sequence $\{k(\ell)\}$, and
  \begin{equation*}
  \frac{1}{L} \norm{\C\f}_{L^2(S)}^2 +Q_1(\f,\f) \geq \frac{1}{c_2
    L^2} \norm{\f}_{L^2(\O_1)}^2
  \end{equation*}
  for $\f\in H^1(\O_1)$.
\end{lem}
\begin{proof}
  By Lemma~2.4, there exist $\m_1,\dots, \m_m>0$ such that
  \begin{equation*}
  \m_1\G_1(x',0)=\m_2\G_2(x',0)=\cdots = \m_m\G_m(x',0).
  \end{equation*}
  We set
  \begin{equation*}
  \G(x)= \m_{k(\ell)}\G_{k(\ell)}(x-\ell\mathbf{e}_d) \quad \text{if
  }\ell\leq x_d\leq \ell+1,
  \end{equation*}
  and then $\G\in H^1(\O_1)$ by the above observation. Moreover, $\G$
  is the ground state for $P_1^N$, unique up to a constant. We apply
  Lemma~2.1 to $\f/\G$, and we have
  \begin{align*}
    \frac{1}{L^2}\norm{\f}_{L^2(\O_1)}^2
    &\leq \frac{1}{L^2}(\sup\G)^2 \norm{\f/\G}^2_{L^2(\O_1)} \\
    &\leq \frac{(\sup\G)^2 }{L}\norm{\C(\f/\G)}_{L^2(S)}^2
    + (\sup\G)^2 \norm{\nabla(\f/\G)}_{L^2(\O_1)}^2 \\
    &\leq
    \biggpare{\frac{\sup\G}{\inf\G}}^2\frac{1}{L}\norm{\C\f}_{L^2(S)}^2
    +c_1 (\sup\G)^2 Q_1(\f,\f),
  \end{align*}
  where we have used Lemma~2.3 in the last inequality. The claim
  follows immediately.
\end{proof}

We next consider $P_0=-\triangle +W_0$ on $L^2(\O_0)$ and its
Dirichlet-to-Neumann operator. As in Theorem~\ref{th:4}, we suppose
\begin{equation*}
\a=\inf \s(P_0^N) >0.
\end{equation*}
We set
\begin{equation*}
P'_0= -\triangle +W_0 \quad\text{on }L^2(\O_0) \text{ with }
\mathcal{D}((P_0')^{1/2})=\bigset{\f\in H^1(\O_0)}{\C\f=0},
\end{equation*}
where $\C$ is the trace operator from $H^1(\O_1)$ to $L^2(S)$.  $P'_0$
defines a self-adjoint operator, and each $\f\in \mathcal{D}(P'_0)$
satisfies Dirichlet boundary conditions on $S$ and Neumann boundary
conditions on $\pa\O_0\setminus S$.  Let $\l<\a$. By a standard
argument of the theory of elliptic boundary value problems (see, e.g.,
Folland \cite{F}), for any $g\in H^{3/2}(S)$, there exists a unique
$\g\in H^2(\O_0)$ such that
\begin{equation}\label{EEq0}
  (-\triangle +W_0-\l)\g =0, \qquad \C\g=g
\end{equation}
and that satisfies Neumann boundary conditions on $\pa\O_0\setminus
S$. Then, the map
\begin{equation*}
T(\l)\ :\ g\mapsto \C(\pa_\n\g)\in H^{1/2}(S)
\end{equation*}
defines a bounded linear map from $H^{3/2}(S)$ to $H^{1/2}(S)$, where
$\pa_\n=\pa/\pa x_d$ is the outer normal derivative on $S$.  We
consider $T(\l)$ as an operator on $L^2(S)$, and it is called the
{\it Dirichlet-to-Neumann operator}.
\begin{lem}
  $T(\l)$ is a symmetric operator. Moreover, if $\l_0<\a$ then
  $T(\l)\geq \e$ for $0\leq\l\leq\l_0$ with some $\e>0$.
\end{lem}
\begin{proof}
  Let $\f,\g\in H^2(\O_0)$ such that $\C\f=f$, $\C\g=g$, and
  \begin{equation*}
  (-\triangle +W_0-\l) \f = (-\triangle +W_0-\l) \g =0
  \end{equation*}
  with Neumann boundary conditions on $\pa \O_0\setminus S$. By
  Green's formula we have
  \begin{align*}
    0 &= \jap{(-\triangle +W_0-\l)\f,\g} - \jap{\f, (-\triangle +W_0-\l) \g} \\
    &= -\int_S \pa_\n\f\cdot \overline\g +\int_S \f\cdot \pa_\n \overline\g
    =-\jap{T(\l)f,g}_{L^2(S)}+\jap{f,T(\l)g}_{L^2(S)},
  \end{align*}
  and hence $T(\l)$ is symmetric. Similarly, we have
  \begin{align*}
    0 &= \jap{(-\triangle +W_0-\l) \f ,\f} \\
    &= -\int_S \pa_\n\f\cdot\overline\f +\int_{\O_0} |\nabla\f|^2
    + \int_{\O_0} (W_0-\l)|\f|^2 \\
    &=-\jap{T(\l)f,f} +Q_0(\f,\f)-\l\norm{\f}^2,
  \end{align*}
  where $Q_0(\f,\f)=\int_{\O_0}\bigpare{|\nabla\f|^2+W_0|\f|^2}dx$.
  Hence, we learn
  \begin{equation*}
  \jap{T(\l)f,f} =Q_2(\f,\f)-\l\norm{\f}^2 \geq
  Q_0(\f,\f)-\l_0\norm{\f}^2.
  \end{equation*}
  The form in the right hand side is equivalent to
  $\norm{\f}^2_{H^1(\O_0)}$ since $\l_0<\a$. Hence, it is bounded from
  below by $\e\norm{f}^2_{L^2(S)}$ with some $\e>0$ by virtue of the
  boundedness of the trace operator from $H^1(\O_0)$ to $L^2(S)$.
\end{proof}

We note that $T(\l)$ extends to a self-adjoint operator on $L^2(S)$ by
the Friedrichs extension, though we do not use the fact in this paper.
\begin{proof}[Proof of Theorem~\ref{th:4}]
  Let $\f$ be the ground state of $P^N$ on $\O$ with the ground state
  energy $\l\geq 0$.  If $\l\geq\l_0>0$ with some fixed $\l_0$
  (independently of $L$), then the statement is obvious, and hence we
  may assume $0\leq\l\leq\l_0<\a=\inf\s(P^N_0)$ without loss of
  generality.

  Let $f=\C\f\in H^{3/2}(S)$. Since $\f$ satisfies Neumann boundary
  conditions on $\pa\O_0\setminus S$, we learn $\pa_\n\f\lceil_S =
  T(\l)\f$.  On the other hand, by Green's formula, we have
  \begin{align*}
    \int_{\O_1} P^N\f\cdot\overline\f &= \int_S
    \pa_n\f\cdot\overline\f
    +\int_{\O_1} |\nabla\f|^2 +W_1|\f|^2 \\
    &=\jap{T(\l)f,f}_{L^2(S)} +Q_1(\f,\f) \\
    &\geq \e\norm{f}^2_{L^2(S)} +Q_1(\f,\f)
  \end{align*}
  by Lemma~2.6. Now, we apply Lemma~2.5 to learn that the right hand
  side is bounded from below by $(1/c_2L^2)\norm{\f}^2_{L^2(\O_1)}$.
  Since $P^N\f=\l\f$ and $\norm{\f}_{L^2(\O_1)}\neq 0$, this implies
  $\l \geq 1/c_2L^2$ for sufficiently large $L$.
\end{proof}


\section{Proof of main theorems}
\label{sec:proof-main-theorems}
%
%
\begin{floatingfigure}{.35\textwidth}
  \begin{center}
    \includegraphics[width=.35\textwidth]{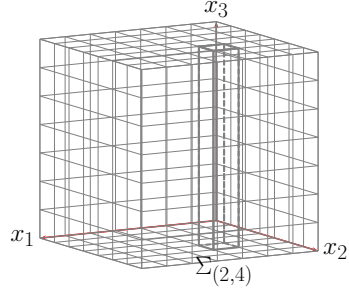}
  \end{center}
  \caption{Chopping the cube into strips}\label{fig:1}
\end{floatingfigure}
%
Here, we mainly discuss the proof of Theorems~\ref{th:1}
and~~\ref{th:2}, and we prove Theorem~\ref{th:3} at the end of the
section. We thus suppose Assumption~A with either $m<M$ or that there
exists $k,k'$ such that $v_k\underset{d}{\not\sim} v_{k'}$.

For notational simplicity, we assume the reflections of $v_k$ at
$\{x_d=1/2\}$ are included in the possible set of potentials
$\{v_k\}$. This does not change the conditions on $\{v_1,\dots,v_m\}$,
but we might need to add the reflections of $\{v_{m+1},\dots,
v_M\}$. This does not affect the following arguments.

We write
\begin{equation*}
\Lambda =\bigset{p\in\ze^{d-1}}{0\leq p_j\leq L-1, j=1,\dots,d-1}
\end{equation*}
and, for $p\in\Lambda$, we set
\begin{equation*}
\Sigma_p =\bigcup_{k=0}^{L-1} C_1((p,k))
\end{equation*}
so that $C_L(0)$ is decomposed (see Fig.~\ref{fig:1}) as
\begin{equation*}
C_L(0) =\bigcup_{p\in\Lambda} \Sigma_p
\end{equation*}
which is a disjoint union except for the boundaries of the strips.

For a given $V_\omega$ and $p\in\Lambda$, we consider the restriction of
$H_\omega$ to $\Sigma_p$, i.e.,
\begin{equation*}
\tilde H_p^N =\triangle + V_0 +\sum_{\ell=0}^{L-1} v_{\omega((p,\ell))}
(x-(p,\ell)) \quad \text{on }L^2(\Sigma_p)
\end{equation*}
with Neumann boundary conditions on $\pa \Sigma_p$.  By the standard
Neumann bracketing, we learn
\begin{equation*}
H_{\omega,L}^N \geq \bigoplus_{p\in\Lambda} \tilde H_p^N \quad
\text{on}\quad L^2(C_L(0))\cong\bigoplus_{p\in\Lambda} L^2(\Sigma_p),
\end{equation*}
and hence, in particular,
\begin{equation}\label{Ineq1}
  \inf \s(H_{\omega,L}^N) \geq \min_{p\in\Lambda} \inf\s(\tilde H_p^N).
\end{equation}
Under our assumptions, one of the following holds for each
$p\in\Lambda$:
\begin{description}
  \itemindent-1em
\item[${(a)}_p$:] $\omega((p,\ell))>m$ for some $\ell$, or
  $v_{\omega((p,\ell))} \underset{d}{\not\sim} v_{\omega((p,\ell'))}$ for some
  \newline $\ell,\ell'\in\{0,\dots,L-1\}$.
\item[${(b)}_{p}$:] For all $\ell,\ell'\in\{0,\dots,L-1\}$,
  $\omega((p,\ell))\leq m$ and $v_{\omega((p,\ell))} \underset{d}{\sim}
  v_{\omega((p,\ell'))}$.
\end{description}

We note that the probability of $(b)_p$ to occur is less than
$\m^{-L}$ with some $\m<1$ independent of $L$. Since $\{\omega(\c)\}$
are independent, we have
\begin{equation}
  \label{prob}
  \mathbb{P}\bigpare{(b)_p \text{ holds for some }p\in\Lambda} \leq L^d \m^{-L},
\end{equation}
which is small if $L$ is large. For the moment, then, we suppose
$(a)_p$ holds for all $p\in\Lambda$.

We denote $V^p(x)$ be the potential function of $\tilde H_p^N$ on
$\Sigma_p$.  Let
\begin{equation*}
\hat\Sigma_p = (p+[0,1]^{d-1})\times (\re/(2L\ze))
\end{equation*}
and set $\hat V^p(x) = V^p(x',|x_d|)$ for $x=(x',x_d)\in
(p+[0,1]^{d-1})\times[-L,L)\cong \hat \Sigma_p$, i.e., $\hat V^p$ is
the extension of $\tilde V^p$ by the reflection at $\{x_d=0\}$.  We
note $\hat V^p$ is continuous on $\hat \Sigma_p$.  We now consider
\begin{equation*}
\hat H_p^N =\triangle +\hat V^p \quad \text{on } L^2(\hat \Sigma_p)
\end{equation*}
with Neumann boundary conditions. It is easy to see
\begin{equation}\label{Ineq2}
  \inf\s(\tilde H_p^N) \geq \inf \s(\hat H_p^N).
\end{equation}
In fact, if $\F$ is the ground state of $\tilde H_p^N$, then we extend
$\F$ by reflection to obtain $\hat\F\in H^1(\hat\Sigma_p)$ and we have
\begin{equation*}
\frac{\jap{\hat H_p^N\hat \F,\hat \F}}{\norm{\hat\F}^2}
=\frac{\jap{\tilde H_p^N \F,\F}}{\norm{\F}^2} =\inf\s (\tilde H_p^N)
\end{equation*}
and the claim \eqref{Ineq2} follows by the variational principle.

Since we assume $(a)_p$, $\Sigma_p$ can be decomposed to subsegments
$\Sigma_p= \bigcup_{j=1}^K \Xi_j$ such that each $\Xi_j$ satisfies the
following conditions: We write
\begin{equation*}
\Xi_j = \bigcup_{\ell=0}^\n C_1(p,\k+\ell), \quad \k\in \ze,\ 0\leq
\n< L,
\end{equation*}
and
\begin{equation*}
\hat V^p(x) = v_{\b(\ell)}(x-(p,\ell))\quad \text{for } x\in
C_1(p,\k+\ell), \ \ell\in\{0,\dots, \n\}
\end{equation*}
with $\b(\ell)\in\{1.\dots,M\}$. Then either one of the following
holds
\begin{enumerate}
  \renewcommand{\theenumi}{\roman{enumi}}
  \renewcommand{\labelenumi}{\textbf{{\rm(\theenumi)} }}
\item $\b(0)\in \{m+1,\dots, M\}$; $\b(\ell)\in \{1,\dots, m\}$ for
  $\ell\geq 1$; and $v_{\b(\ell)}\underset{d}{\sim}v_{\b(\ell')}$ for
  $\ell,\ell'\in \{1,\dots,\n\}$.
\item $\b(\ell)\in \{1,\dots, m\}$ for all $\ell$;
  $v_{\b(0)}\underset{d}{\not\sim} v_{\b(1)}$; and
  $v_{\b(\ell)}\underset{d}{\sim}v_{\b(\ell')}$ for $\ell,\ell'\in
  \{2,\dots,\n\}$.
\end{enumerate}

The proof of this claim is an easy combinatorics, though somewhat
lengthy to write down using symbols. We omit the details.

We again decompose $\hat H_p^N$. We denote the restriction of $\hat
H_p^N$ to $\Xi_j$ by $P_j$ on $L^2(\Xi_j)$ with Neumann boundary
conditions.  Then, again by Neumann bracketing, we learn
\begin{equation*}
\hat H_p^N \geq \bigoplus_{j=1}^\k P_j \quad \text{on }
L^2(\hat\Sigma_p)\cong \bigoplus_{j=1}^\k L^2(\Xi_j),
\end{equation*}
and in particular,
\begin{equation}
  \label{Ineq3}
  \inf \s(\hat H_p^N)\geq \min_j \inf \s(P_j).
\end{equation}

Now if (i) holds for $\Xi_j$, then we set $a=1$ and use Theorem~\ref{th:4}
for $P_j$.  Since $\inf \s(H_{\b(0)}^N)>0$ by Assumption~A and $\n\leq
L$, we learn
\begin{equation*}
\inf \s(P_j)\geq \frac{1}{C(\n-1)^2}\geq \frac{1}{C(L-1)^2}.
\end{equation*}
If (ii) holds for $\Xi_j$, then we set $a=2$ and use Theorem~\ref{th:4} for
$P_j$.  Since $v_{\b(0)}\underset{d}{\not\sim}v_{\b(1)}$, we have
$\inf\s(H_{\b(0)\b(1)(d)}^N)>0$. Thus we have
\begin{equation*}
\inf \s(P_j)\geq \frac{1}{C(\n-2)^2} \geq \frac{1}{C(L-2)^2}.
\end{equation*}
Combining these with \eqref{Ineq1}, \eqref{Ineq2} and \eqref{Ineq3},
we conclude
\begin{equation}\label{Ineq4}
  \inf \s(H_{\omega,L}^N)\geq \frac{c_3}{L^2}
\end{equation}
with some $c_3>0$, provided $(a)_p$ holds for all $p\in\Lambda$.
\begin{proof}[Proof of Theorems~\ref{th:1} and~\ref{th:2}]
  For $E>0$, we set
  \begin{equation*}
  \sqrt{\frac{c_3}{E}} <L \leq \sqrt{\frac{c_3}{E}}+1
  \end{equation*}
  so that, by virtue of \eqref{Ineq4},
  \begin{equation*}
  \inf \s(H_{\omega,L}^N)>E
  \end{equation*}
  provided Condition $(a)_p$ holds for all $p\in\Lambda$. As noted in
  \eqref{prob}, the probability of the events that $(b)_p$ holds for
  some $p\in\Lambda$ is bounded by
  \begin{equation*}
  \mathbb{P}\bigpare{(b)_p \text{ for some }p\in \Lambda} \leq L^d
  \m^{-L} \leq c_4 E^{-d/2} e^{-c_5E^{-1/2}}
  \end{equation*}
  with some $c_4,c_5>0$. On the other hand, since the potential
  $V_0+V_\omega$ is uniformly bounded, we have
  \begin{equation*}
  \#\{\text{eigenvalues of }H_{\omega,L}^N\leq \a\}\leq c_6 L^d
  \end{equation*}
  for any $\omega$ with some $c_6>0$. Thus we have
  \begin{align*}
    &L^{-d}\mathbb{E}\bigpare{\#\{\text{e.v. of }H^N_{\omega,L}\leq E\}}
    \leq L^{-d} (c_6 L^d) \mathbb{P}\bigpare{(b)_p \text{for some }p\in\Lambda} \\
    &\hspace{2cm} \leq c_4 c_6 E^{-d/2} e^{-c_5 E^{-1/2}} \leq c_7
    e^{-(c_5-\e)E^{-1/2}}
  \end{align*}
  for $0<\e<c_5$ with some $c_7>0$. By the Neumann bracketing again,
  we have
  \begin{equation*}
  N(E)\leq L^{-d}\mathbb{E}\bigpare{\#\{\text{e.v. of }H^N_{\omega,L}\leq
    E\}} \leq c_7 e^{-(c_5-\e)E^{-1/2}}
  \end{equation*}
  and Theorems~\ref{th:1} and~\ref{th:2} follow immediately from this
  estimate.
\end{proof}

In fact, we have proved
\begin{equation*}
\liminf_{E\to+0} \frac{|\log N(E)|}{E^{-1/2}} >0,
\end{equation*}
and this statement is slightly stronger than~\eqref{LT}.
\begin{proof}[Proof of Theorem~\ref{th:3}]
  (i) This statement is an immediate consequence of Assumption~B and
  Theorem~\ref{th:2}.  We just replace the $x_d$-axis by the $x_j$-axis where
  $v_k\underset{j}{\not\sim}v_\ell$ for some $k,\ell$.

  (ii) We use the ground state transform as in the proof of Lemmas
  2.3--2.5.  Under our conditions, there exist $\m_1,\dots,\m_m>0$
  such that
  \begin{equation*}
  \m_1\G_1(x)=\m_2\G_2(x)=\cdots=\m_m\G_m(x) \quad \text{for }x\in\pa
  C_1(0).
  \end{equation*}
  For given $H_{\omega,L}^N$, we set
  \begin{equation*}
  \F(x)= \m_k\G_k(x) \quad \text{if }x\in C_1(\c) \text{ with }
  \omega(\c)=k.
  \end{equation*}
  Then it is easy to see that $\F$ is the positive ground state of
  $H_{\omega,L}^N$ with the energy 0.  Let $Q(\cdot,\cdot)$ be the
  quadratic form corresponding to $H_{\omega,L}^N$.  For $\f\in
  H^1(C_L(0))$, we set $f=\f/\F$. As in the proof of Lemma~2.3, we
  have
  \begin{equation*}
  Q(\f,\f) = \norm{\F(\nabla f)}^2
  \end{equation*}
  and hence
  \begin{equation*}
  (\inf \F)^2 \norm{\nabla f}^2 \leq Q(\f,\f)\leq (\sup\F)^2
  \norm{\nabla f}^2.
  \end{equation*}
  This implies
  \begin{equation*}
  K^{-2} \frac{\norm{\nabla f}^2}{\norm{f}^2} \leq
  \frac{Q(\f,\f)}{\norm{\f}^2} \leq K^2 \frac{\norm{\nabla
      f}^2}{\norm{f}^2}
  \end{equation*}
  where $K= \max_k \sup (\m_k\G_k) / \min_k \inf (\m_k \G_k)$.  By the
  min-max principle, we learn
  \begin{align*}
    K^{-2} \#\{\text{e.v. of }(-\triangle)_L^N\leq E\}
    &\leq \#\{\text{e.v. of }H_{\omega,L}^N\leq E\} \\
    &\hspace{1cm}\leq K^2 \#\{\text{e.v. of }(-\triangle)_L^N\leq E\}
  \end{align*}
  where $(-\triangle)_L^N$ is the Laplacian on $C_L(0)$ with Neumann
  boundary conditions. Taking the limit $L\to+\infty$, we have
  \begin{equation}
    \label{eq:4}
    K^{-2} c_d E^{d/2} \leq N(E)\leq K^2 c_d E^{d/2},    
  \end{equation}
  where $c_d$ is the volume of the unit ball in $\re^d$. This
  completes the proof of Theorem~\ref{th:3}. 
  \end{proof}


\section{Application to random displacement models}
\label{sec:appl-rand-displ}
We now consider a model recently studied by Baker, Loss and Stolz
in~\cite{BLS1, BLS2}. Combining their results with Theorem~\ref{th:1},
we show that this model exhibits Lifshitz singularities at the ground
state energy.

We consider a random Schr{\"o}dinger operator of the form:
\begin{equation*}
H_\omega=-\triangle +V_\omega \quad\text{on } L^2(\re^d)
\end{equation*}
where
\begin{equation*}
V_\omega(x)=\sum_{\c\in\ze^d} q(x-\c-\omega(\c))
\end{equation*}
with i.i.d.\ random variables $\{\omega(\c)\,|\,\c\in\ze^d\}$ which take
values in $C_1(0)$.
\begin{ass}
  (1) There exists $\d\in(0,1/2)$ such that $\omega(\c)$ takes values
  in a finite set
  \begin{equation*}
    \Theta \subset \bigset{x\in\re^d}{\d\leq x_j\leq 1-\d,\ \forall
      j\in\{1,\dots,d\}}.
  \end{equation*}
  Moreover
  \begin{equation*}
    \Theta\supset \Delta = \bigset{x\in\re^d}{x_j=\d\text{ or }1-\d,\
      \forall j\in\{1,\dots,d\}}
  \end{equation*}
  and $\mathbb{P}(\omega(\c)=x)>0$ for $x\in \Delta$. \newline (2)
  $q\in C_0(\re^d)$ and it is supported in $\bigset{x}{|x_j|\leq \d,
    j\in\{1,\dots,d\}}$.  Moreover, $q$ is symmetric about
  $\{x\,|\,x_j=0\}$, $j=1,\dots, d$. \newline (3) Let
  $H_q^N=-\triangle +q$ on $L^2(\{|x|\leq 1\})$ with Neumann boundary
  conditions, and let $\phi$ be the ground state. Then, $\phi$ is not
  a constant outside $\mbox{Supp }q$. Note that this is relevant only
  if the ground state energy is 0.
\end{ass}
%
\begin{figure}[h]
  \begin{center}
    \subfigure[A typical random configuration]{\includegraphics[width=.3\textwidth]{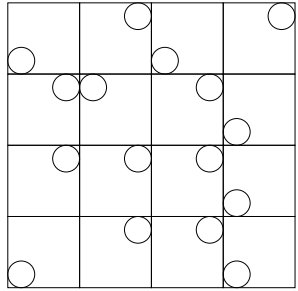}}\hskip3cm
    \subfigure[The minimizing configuration]{\includegraphics[width=.3\textwidth]{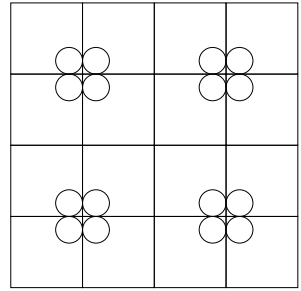}}
  \end{center}
  \caption{An example in two dimensions }\label{fig:2}
\end{figure}
%
Let $H_{1,\b}^N=-\triangle +q(x-\b)$ on $L^2(C_1(0))$ with Neumann
boundary conditions, where $\b\in\Theta$. Baker, Loss and Stolz
\cite{BLS1} showed that $\inf\s(H_{1,\b}^N)$ takes its minimum (with
respect to $\b$) if and only if $\b\in \Delta$.\\
In particular, they showed that for $H^N_{\omega,2\ell}$ the Neumann
restriction of $H_\omega$ to $C_{2\ell}(0)$ the minimal value of the
ground state energy was obtained for clustered configuration (see
Fig~\ref{fig:2}).
%
\begin{figure}[h]
  \begin{center}
    \subfigure[The minimal $2\times 2$ configurations]{\includegraphics[width=.45\textwidth]{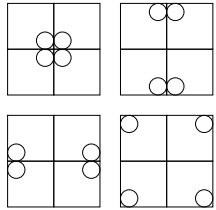}}\hskip1cm
    \subfigure[Other  $2\times 2$ configurations]{\includegraphics[width=.45\textwidth]{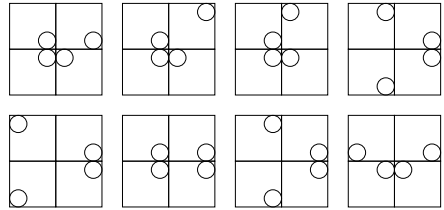}}
  \end{center}
  \caption{$2\times 2$ configurations in two dimensions}\label{fig:3}
\end{figure}
%

We cannot directly apply our result to this model, since $q(x-\b)$ is
not symmetric for $\b\in\Delta$. However, they also showed that if we
consider the operator $H_\omega$ restricted to $C_2(0)$ and if $d\geq 2$,
then the minimum is attained by $2^d$ symmetric configurations, which
are equivalent to each other by translations (see~\cite{BLS2} and
Fig.~\ref{fig:3}). Thus, we can apply our results by considering
$H_\omega$ as a $2\ze^d$-ergodic random Schr{\"o}dinger operators, i.e., by
considering $C_2(0)$ as the unit cell.  Then, this model satisfies
Assumption~A with $M=(\#\Theta)^{2^d}$ and $m=2^d$.
\begin{thm}
  \label{th:5}
  Let $d\geq 2$, and suppose Assumption~C for some $\d\in(0,1/2)$.
  Then,~\eqref{LT} holds at the bottom of the spectrum of $H_\omega$, a.s.
\end{thm}
\noindent We note that if $d=1$, this result does not hold, and the
IDS may have logarithmic singularity at the bottom of the spectrum
(\cite{BLS2}). In view of our results, such singularities can occur
for the lack of symmetry of the minimizing configurations.


%
\section{The alloy type model studied in~ \cite{KN1}}
\label{sec:alloy-type-model}
In a previous paper on Lifshitz tails for sign indefinite alloy type
random Schr{\"o}dinger operators~\cite{KN1}, we studied the
model~\eqref{eq:5} for a single site potential $V$ satisfying the
reflection symmetry Assumption B.\\
We now recall some of the results of that work. Let the support of the
random variables $(\omega_\gamma)_\gamma$ be contained in $[a,b]$ and
assume both $a$ and $b$ belong to the essential support of the random
variables.\\
Consider now the operator $H_\lambda^N=-\Delta+\lambda V$ with Neumann
boundary conditions on the cube $C_1(0)=[0,1]^d$. Its spectrum is
discrete, and we let $E_-(\lambda)$ be its ground state energy. It is
a simple eigenvalue and $\lambda\mapsto E_-(\lambda)$ is a real
analytic concave function defined on $\R$. Let $E_-$ be the infimum of
the almost sure spectrum of $H_\omega$ then
\begin{Pro}[\cite{KN1}]
  \label{pro:1}
  Under Assumption B,
  \begin{equation*}
    E_-=\inf(E_-(a),E_-(b)).
  \end{equation*}
\end{Pro}
\noindent As for Lifshitz tails, we proved
\begin{thm}[\cite{KN1}]
  \label{thr:2}
  Suppose Assumption B is satisfied. Assume moreover that
  \begin{equation}
    \label{eq:11}
    E_-(a)\not=E_-(b).
  \end{equation}
  Then
  \begin{equation*}
    \limsup_{E\to E_-^+}\frac{\log|\log N(E)|}{\log(E-E_-)}\leq
    -\frac d2-\alpha_+
  \end{equation*}
  where we have set $c=a$ if $E_-(a)<E_-(b)$ and $c=b$ if
  $E_-(a)>E_-(b)$, and
  \begin{equation*}
    \alpha_+=-\frac12\liminf_{\varepsilon\to0}\frac{\log|\log\mathbb{P}(\{
      |c-\omega_0|\leq\varepsilon\})|}{\log\varepsilon}\geq0.
  \end{equation*}
\end{thm}
\noindent The technique developed in~\cite{KN1} did not allow us to
treat the case $E_-(a)=E_-(b)$. Clearly, if the random variables
$(\omega_\gamma)_\gamma$ are non trivial and Bernoulli distributed,
i.e., if $\mathbb{P}(\omega_0=a)+\mathbb{P}(\omega_0=b)=1$ and
$\mathbb{P}(\omega_0=a)>0$, $\mathbb{P}(\omega_0=b)>0$,
Theorem~\ref{th:3} tells us that the Lifshitz tails hold if and only
if $aV\underset{j}{\not\sim}bV$ for some $j\in\{1,\cdots,d\}$
(see~\eqref{eq:14}). So we are just left with the case when the random
variables $(\omega_\gamma)_\gamma$ are not Bernoulli distributed.\\
We prove
\begin{thm}
  \label{thr:3}
  Suppose assumption B is satisfied and that 
  \begin{equation}
    \label{eq:11}
    E_-(a)=E_-(b).
  \end{equation}
  Assume moreover that the i.i.d. random variables
  $(\omega_\gamma)_\gamma$ are not Bernoulli distributed
  i.e. $\mathbb{P}(\omega_0=a)+\mathbb{P}(\omega_0=b)<1$.\\
  Then
  \begin{equation}
    \label{eq:7}
    \limsup_{E\to E_-^+}\frac{\log|\log N(E)|}{\log(E-E_-)}\leq
    -\frac 12.
  \end{equation}
\end{thm}
\noindent So we show that Lifshitz tails also hold in this case. As
already noted we believe that~\eqref{eq:7} is not optimal and that
$-1/2$ should be replaced by $-d/2$. Moreover, depending on the tail
of the distributions of the random variables $(\omega_\gamma)_\gamma$
near $a$ and $b$, the $\limsup$ in~(\ref{eq:7}) should be a limit, the
inequality should become an equality, the exponent $-1/2$ should be
replaced by $-d/2$ plus a possibly vanishing constant (see Section 0
of~\cite{KN1} for the case $E_-(a)\not=E_-(b)$). \\
\noindent Combining Theorems~\ref{thr:2} and~\ref{thr:3} with the
Wegner estimates obtained in~\cite{Kl:95b,Hi-Kl:01a} and the
multiscale analysis as developed in~\cite{MR2002m:82035}, we learn
\begin{thm}
  \label{thr:1}
  Assume Assumption B. Assume, moreover, that the common distribution
  of the random variables admits an absolutely continuous
  density. Then, the bottom edge of the spectrum of $H_\omega$
  exhibits complete localization in the sense of~\cite{MR2002m:82035}.
\end{thm}
\noindent This result improves upon Theorem~0.3 of~\cite{KN1}.
\subsection{The proof of Theorem~\ref{thr:3}}
\label{sec:proof-theorem}
Recall that $H_{\omega,L}^N$ is defined in~(\ref{eq:9}). It is well
known that, at $E$, a continuity point of $N(E)$, the sequence
\begin{equation*}
  N_L^N(E)=\mathbb{E}\left(\frac{\#\{\text{eigenvalues of
      }H_{\omega,L}^N\ \leq E\}}{L^d}\right)
\end{equation*}
is decreasing and converges to $N(E)$ (see e.g.~\cite{FP,K1}). As
\begin{equation}
  \label{eq:7}
  N_L^N(E)\leq
  C\,\mathbb{P}(\{\inf\sigma(H_{\omega,L}^N)\leq E\})
\end{equation}
it is sufficient to prove an upper bound for
$\mathbb{P}(\{\inf\sigma(H_{\omega,L}^N)\leq E\})$ for a well chosen
value of $L$.\\
Define $E_{-,L}(\omega)=\inf\sigma(H_{\omega,L}^N)$. It only depends
on $(\omega_\gamma)_{\gamma\in Z_L}$, where 
\begin{equation*}
Z_L=\bigset{\c\in\ze^d}{0\leq\c_j<L, j=1,\dots,d}.
\end{equation*}
One has
\begin{Le}
  \label{le:2}
  The function $\omega\mapsto E_{-,L}(\omega)$ is real analytic and
  strictly concave on $[a,b]^{Z_L}$.
\end{Le}
\begin{proof}
  Though this is certainly a well known result, for the sake of
  completeness, we give the proof. The ground state being
  simple, $\omega\mapsto E_{-,L}(\omega)$ is real analytic in $\omega$.\\
  As $H_\omega$ depends affinely on $\omega$, by the variational
  characterization of the ground state energy, $E_{-,L}(\omega)$ is the
  infimum of a family of affine functions of $\omega$. So it is concave.\\
  The strict concavity is obtained using perturbation theory. Let
  $\varphi_L(\omega)$ be the unique normalized positive ground state
  associated to $E_{-,L}(\omega)$ and $H_{\omega,L}^N$. The ground
  state energy being simple, this ground state is a real analytic
  function of $\omega$; differentiating once the eigenvalue equation
  and the normalization condition of the ground state, as the ground
  state is normalized and real, one obtains
  \begin{equation}
    \label{eq:10}
    (H_{\omega,L}^N-E_{-,L}(\omega))\partial_{\omega_\gamma}\varphi_L(\omega)
    =\left(\partial_{\omega_\gamma}E_{-,L}(\omega)-V(\cdot-\gamma)\right)
    \varphi_L(\omega)
  \end{equation}
  and
  \begin{equation}
    \label{eq:13}
    \langle\partial_{\omega_\gamma}\varphi_L(\omega),\varphi_L(\omega)\rangle=0.
  \end{equation}
  A second differentiation yields
  \begin{equation*}
    \begin{split}
      (H_{\omega,L}^N-E_{-,L}(\omega))\partial^2_{\omega_\gamma\omega_\beta}
      \varphi_L(\omega)
      &=\partial^2_{\omega_\gamma\omega_\beta}E_{-,L}(\omega)\varphi_L(\omega)
      \\
      &\hskip.5cm+\left(\partial_{\omega_\gamma}E_{-,L}(\omega)-V(\cdot-\gamma)\right)
      \partial_{\omega_\beta}\varphi_L(\omega)\\
      &\hskip1cm+\left(\partial_{\omega_\beta}E_{-,L}(\omega)-V(\cdot-\beta)\right)
      \partial_{\omega_\gamma}\varphi_L(\omega).
    \end{split}
  \end{equation*}
  Hence, using~\eqref{eq:10} and~\eqref{eq:13}, we compute
  \begin{equation*}
    \begin{split}
      \partial^2_{\omega_\gamma\omega_\beta}E_{-,L}(\omega)&= -\langle
      V(\cdot-\gamma)\partial_{\omega_\beta}\varphi_L(\omega),
      \varphi_L(\omega)\rangle -\langle
      V(\cdot-\beta)\partial_{\omega_\gamma}\varphi_L(\omega),
      \varphi_L(\omega)\rangle\\
      &= -2\text{Re}
      \left(\left\langle(H_{\omega,L}^N-E_{-,L}(\omega))^{-1}
          \psi_\beta,\psi_\gamma \right\rangle\right)
    \end{split}
  \end{equation*}
  where
  \begin{itemize}
  \item $\psi_\gamma= \Pi V(\cdot-\gamma)\varphi_L(\omega)$
  \item $\Pi$ is the orthogonal projector on the orthogonal to
    $\varphi_L(\omega)$.
  \end{itemize}
  Hence, for $(a_\gamma)_\gamma$ complex numbers,
  \begin{equation*}
    \sum_{\gamma,\beta}\partial^2_{\omega_\gamma\omega_\beta}
    E_{-,L}(\omega)a_\gamma\overline{a_\beta}
    =-2\text{Re} \left(\left\langle(H_{\omega,L}^N-E_{-,L}(\omega))^{-1}
        \Pi u_a ,\Pi u_a\right\rangle\right)
  \end{equation*}
  where $u_a=(\sum_\gamma a_\gamma V(\cdot-\gamma))\varphi_L(\omega)$.
  Note that, as $V$ is not trivial, the assumption $E_-(a)=E_-(b)$
  implies that $V$ changes sign, i.e., there exists $x_+\not=x_-$ such
  that $V(x_-)\cdot V(x_+)<0$. Now, the vector $\Pi u_a$ vanishes if and
  only if $u_a$ is colinear to $\varphi_L(\omega)$ which cannot happen
  as $V$ is not constant and $\varphi_L(\omega)$ does not vanish on
  open sets by the unique continuation principle. On the other hand,
  $E_{-,L}(\omega)$ being a simple eigenvalue associated to
  $\varphi_L(\omega)$,
  $\Pi(H_{\omega,L}^N-E_{-,L}(\omega))^{-1}\Pi\geq c\,\Pi$ for some
  $c>0$. So the Hessian of $\omega\mapsto E_{-,L}(\omega)$ is positive
  definite. This completes the proof of Lemma~\ref{le:2}.
\end{proof}
We now turn to the proof of Theorem~\ref{thr:3}.  As the random
variables are not Bernoulli distributed, i.e.,
$\mathbb{P}(\omega_0=a)+\mathbb{P}(\omega_0=b)<1$, we can fix
$\varepsilon>0$ sufficiently small such that
$\mathbb{P}(\omega_0\in[a,a+\varepsilon))
+\mathbb{P}(\omega_0\in(b-\varepsilon,b])<1$. By strict concavity of
$E_-(\lambda)$, one has $E_-(a)<E_-(a+\varepsilon)$ and
$E_-(b)<E_-(b-\varepsilon)$.\\
In Section~\ref{sec:lower-bounds-ground}, we have proved
\begin{Le}
  \label{le:14}
  Assume $E_-(a)=E_-(b)$. There exists $C>0$ such, for all $L\geq0$,
  if $\omega\in\{a,b,a+\varepsilon,b-\varepsilon\}^{Z_L}$ is such
  that
  \begin{description}
  \item[(P)] for all $p\in\Lambda$, there exists $\ell\in \{0,\dots, L-1\}$ 
  such that
  \begin{equation*}
    \omega_{(p,\ell)}\in\{a+\varepsilon,b-\varepsilon\}
  \end{equation*}

  \end{description}
  then
  \begin{equation}
    \label{eq:15}
    E_{-,L}(\omega)\geq E_-(a)+\frac1{CL^2}.
  \end{equation}
\end{Le}
\noindent To complete the proof of Theorem~\ref{thr:3}, we first
extend lemma~\ref{le:14} using the concavity of the ground state
energy to
\begin{Le}
  \label{le:15}
  Assume $E_-(a)=E_-(b)$. 
  There exists $C>0$ such, for all $L\geq0$, if $\omega\in\Omega_L$ is
  such that
  \begin{description}
  \item[(P')] for all $p\in\Lambda$, there exists $\ell\in \{0,\dots, L-1\}$ 
  such that
  \begin{equation*}
    \omega_{(p,\ell)}\in[a+\varepsilon,b-\varepsilon]
  \end{equation*}

  \end{description}
  then~\eqref{eq:15} holds (with the same constant as in
  Lemma~\ref{le:14}).
\end{Le}
\noindent Let us postpone the proof of this result to complete that of
Theorem~\ref{thr:3}. Pick $E> E_-(a)=E_-(b)$. We use~\eqref{eq:7} and
pick $L=c(E-E_-(a))^{1/2}$.  Pick $c>0$ sufficiently small that
$Cc^2<1$.  Then, Lemma~\ref{le:14} tells us that, if
$\omega\in[a,b]^{Z_L}$ satisfies (P'), then $E_-(\omega)> E$. So,
the set $\Omega_L(E):=\{\omega\in\Omega_L;\ E_-(\omega)> E\}$
satisfies
\begin{equation*}
  \Omega_L\setminus\Omega_L(E)\subset \{\omega\in\Omega_L;\
  \exists p\in  \Lambda,\text{ s.t. }\forall\ell, \ 
  \omega_{(p,\ell)}\in[a,a+\varepsilon)\cup(b-\varepsilon,b]\}.
\end{equation*}
Hence,
\begin{equation*}
  \begin{split}
    \mathbb{P}(\Omega_L\setminus\Omega_L(E))&\leq \sum_{p\in \Lambda}
    \mathbb{P}(\{\omega_{(p,\ell)}\in[a,a+\varepsilon)\cup(b-\varepsilon,b]
    \text{ for }\forall\ell\})\\
    &=L^{d-1} \bigpare{\mathbb{P}(\omega_0\in[a,a+\varepsilon))
+\mathbb{P}(\omega_0\in(b-\varepsilon,b])}^{L}
  \end{split}
\end{equation*}
This yields the announced exponential decay and completes the proof of
Theorem~\ref{thr:3}.\qed
\begin{proof}[Proof of Lemma~\ref{le:15}]
  We will proceed in two steps. First, we prove that, if $\omega$
  satisfies (P') and all its coordinates that are not in
  $[a+\varepsilon,b-\varepsilon]$ are either equal to $a$ or to $b$,
  then~\eqref{eq:15} holds (with the same constant as in
  Lemma~\ref{le:14}). This comes from the concavity of the ground
  state and the fact that any such point is a convex combination of
  points satisfying (P). Indeed, take such a point $\omega$ and let
  $\Gamma(\omega)$ be the set of coordinates such that
  $\omega_\gamma\in[a+\varepsilon,b-\varepsilon]$. Define
  $K(\omega)=\{a+\varepsilon,b-\varepsilon\}^{\Gamma(\omega)}$. Then,
  there exists a convex combination $(\mu_\eta)_{\eta\in K(\omega)}$
  such that
  \begin{equation*}
     (\omega_\gamma)_{\gamma\in\Gamma(\omega)}=\sum_{\eta\in
       K(\omega)}\mu_\eta \eta,\quad\quad \sum_{\eta\in
       K(\omega)}\mu_\eta=1,\quad \mu_\eta\geq0.
  \end{equation*}
  Hence,
  \begin{equation*}
    \omega=\sum_{\eta\in K(\omega)}\mu_\eta \tilde\eta
  \text{ where }
    (\tilde\eta)_\gamma=
    \begin{cases}
      \eta_\gamma&\text{ if
      }\gamma\in\Gamma(\omega),\\\omega_\gamma&\text{ if
      }\gamma\not\in\Gamma(\omega).
    \end{cases}
  \end{equation*}
  That $\omega$ satisfies~\eqref{eq:15} then follows from the
  concavity of $\omega\mapsto E_{-,L}(\omega)$, that is
  Lemma~\ref{le:2}, and from Lemma~\ref{le:14}. \\
  To complete the proof of Lemma~\ref{le:15}, it suffices to show that
  a point $\omega$ satisfying (P') can be written a convex combination
  of points of the type defined above. This is done as above. Indeed,
  pick $\omega$ satisfying (P'). Define
  $L(\omega)=\{a,b\}^{(Z_L\setminus\Gamma(\omega))}$. Then, there
  exists a convex combination $(\mu_\eta)_{\eta\in L(\omega)}$ such
  that
  \begin{equation*}
     (\omega_\gamma)_{\gamma\in(Z_L\setminus\Gamma(\omega))}
     =\sum_{\eta\in L(\omega)}\mu_\eta \eta,\quad \sum_{\eta\in
       L(\omega)}\mu_\eta=1,\quad \mu_\eta\geq0.
  \end{equation*}
  Hence,
  \begin{equation*}
    \omega=\sum_{\eta\in L(\omega)}\mu_\eta \tilde\eta
    \text{ where }
    (\tilde\eta)_\gamma=
    \begin{cases}
      \eta_\gamma&\text{ if
      }\gamma\not\in\Gamma(\omega),\\\omega_\gamma&\text{ if
      }\gamma\in\Gamma(\omega).
    \end{cases}
  \end{equation*}
  That $\omega$ satisfies~\eqref{eq:15} then follows from the
  concavity of $\omega\mapsto E_{-,L}(\omega)$ and from the first
  step. This completes the proof of Lemma~\ref{le:15}.
\end{proof}
%

%
%

\end{document}